\documentclass[pra,twocolumn,showpacs,superscriptaddress]{revtex4-1}
\usepackage{amsmath,amscd,amssymb,color,amsthm}
\usepackage{graphicx,amsfonts,dsfont}
\usepackage{epstopdf}
\usepackage{hyperref}
\usepackage{enumerate,bbold}
\usepackage{array, comment, bm, braket}
\newtheorem{prop}{Proposition}
\newcommand{\ie}{\textit{i.e.}}
\usepackage{color}

\begin{document}
\title{Monogamy relations of entropic non-contextual inequalities and 
their experimental realization}
\author{Dileep Singh}
\email{dileepsingh@iisermohali.ac.in}
\affiliation{Department of Physical Sciences, Indian
Institute of Science Education \& 
Research Mohali, Sector 81 SAS Nagar, 
Manauli PO 140306 Punjab, India.}
\author{Jaskaran Singh}
\email{jsinghiiser@gmail.com}
\affiliation{Departamento de F\'{\i}sica Aplicada II, Universidad de Sevilla, E-41012 Sevilla, Spain}
\author{Kavita Dorai}
\email{kavita@iisermohali.ac.in}
\affiliation{Department of Physical Sciences, Indian
Institute of Science Education \& 
Research Mohali, Sector 81 SAS Nagar, 
Manauli PO 140306 Punjab, India.}
\author{Arvind}
\email{arvind@iisermohali.ac.in}
\affiliation{Department of Physical Sciences, Indian
Institute of Science Education \& 
Research Mohali, Sector 81 SAS Nagar, 
Manauli PO 140306 Punjab, India.}
\affiliation{Vice Chancellor, Punjabi University Patiala,
147002, Punjab, India}
\begin{abstract}
We develop a theoretical framework based on a
graph theoretic approach to analyze monogamous relationships
of entropic non-contextuality (ENC) inequalities.  While ENC
inequalities are important in quantum information
theory and are well studied, theoretical as well as
experimental demonstration of their monogamous nature is
still elusive.  We provide conditions for ENC inequalities
to exhibit a monogamous relationship and derive the same for
general scenarios.  We show that two entropic versions of
the Bell-CHSH inequality acting on a tripartite scenario
exhibit a monogamous relationship, for which we provide a
theoretical proof as well as an experimental validation on
an NMR quantum information processor.  
Our experimental technique to evaluate entropies has
been designed to obtain information about entropies via
measurement of only the expectation values of observables.
\end{abstract} 
\maketitle 

\section{Introduction}
\label{sec:intro}
It is well known that measurements on quantum states exhibit
correlations which defy understanding based on classical theories. The outcomes of
measurements performed on spatially separated or even on
single indivisible systems may exhibit correlations greater
than their apparent classical values. These correlations are
termed as Bell
non-local~\cite{bell_non_locality,chsh,bell_hidd_var,bell_high_dim}
or
contextual~\cite{ks_theorem,acin2015,nc_n_cycle,e_principle_corr,graph_cabello}
for bipartite space-like separated and for single
indivisible systems respectively, and are studied by means of
inequalities, a violation of which implies non-classical
behavior~\cite{bell_test2,bell_test4,bell_test5,context_test1,context_test2}.
The most widely studied of these inequalities include the
Bell-CHSH~\cite{chsh} and the KCBS inequality~\cite{kcbs},
quantifying Bell non-local and contextual correlations, respectively.

Quantum correlations that
go beyond the classical allowed values have found immense
application in quantum information processing tasks,
including self-testing~\cite{self_test_pr,self_test2},
quantum key distribution
~\cite{key_dist_review1,key_dist_review2}, device
independent quantum key
distribution~\cite{Acin_sec,diqkd_vazirani} and randomness
certification~\cite{randomness_acin}, to name a few. A common
ground of most of these applications is that they utilize
the property of monogamy of correlations \cite{acin-prl-06,
pawloski-pra-10, saha-pra-17, li-osa-18}.  Such monogamy
relations dictate that when two parties are violating a
particular inequality (contextual or Bell type), then it is not
possible for one of the parties to violate a second
similar
inequality with a third
party on a
tripartite state shared amongst all three of them. These monogamy
relations have been well studied in the case of Bell
inequalities and
contextuality~\cite{pawlowski-prl-09,ramanathan-prl-12}
Interestingly, monogamous relations also exist between Bell
non-locality and contextuality and these have been verified
experimentally~\cite{kurzy-prl-14, zhan-prl-16,
saha-pra-17}.

Offering a different perspective on quantum correlations,
Braunstein and Caves introduced an information theoretic
approach to understand non-local
correlations~\cite{bell_info_caves,entropic_no_signalling}
studied in the form of inequalities. These inequalities are
eponymously termed as entropic inequalities.  In this
formalism, the joint Shannon entropies carried by
observables must satisfy an inequality in order to have a
local description of correlations. This makes the
inequalities non-linear functions of probabilities, unlike
standard Bell inequalities which are linear. These inequalities 
are one-way
statements;
only when an inequality is violated, can it be implied that
a joint probability distribution over the observables under
consideration cannot exist. However, if the inequality is
satisfied, no definite conclusion can be
drawn. Subsequently,
the work was extended to include non-contextual
scenarios~\cite{entropic_chaves,entropic_fritz} as well, and
the inequalities are generally termed as entropic
non-contextuality (ENC) inequalities. Since their inception,
these inequalities have been studied extensively and have
also been realized experimentally~\cite{entropic_test,
katiyar-pra-13, cao-sr-16, zhan-prl-17, qu-pra-20}.
Experimental implementations of state
dependent~\cite{zhan-prl-17} as well as state
independent~\cite{qu-pra-20} ENC inequalities have been
performed on optical systems. 

These ENC inequalities offer an added advantage of being
independent of the number of outcomes in a measurement,
unlike Bell scenarios where the number of outcomes of a
measurement play an important rule in determining the local value.
This property of ENC inequalities makes them suitable candidates for
applications in non-locality
distillation~\cite{entropic_fritz}  and bilocality
scenarios, among many other applications. 
Due to their non-linear nature, a
general description of monogamous relations for all possible
ENC scenarios is a nontrivial
task.  While there exist scenarios where it is possible to derive a monogamous
relationship~\cite{entropic_no_signalling}, a
generalization of this recipe would be too complex and
require involved calculations. There have been attempts in this direction~\cite{kurzy-pra-14} but based
on a different approach than ours.  Furthermore, while there are a few
experimental studies verifying monogamy of ENC inequalities,
all of them are based on optical methods, which are capable
of directly yielding probabilities of outcomes of a
measurement in the form of frequency of photon
detections.  However, to the best of our knowledge, there is
currently no experimental verification of monogamy of ENC
inequalities till date.

Providing a fresh perspective
in this paper, we describe an elegant theoretical
recipe to evaluate monogamous relations of ENC inequalities
and verify the results on an NMR quantum information
processor. We develop a
theoretical framework, based on graph theoretic formalism,
to study monogamous relations of ENC inequalities in
arbitrary no-signalling scenarios, apply our formalism to
the entropic Bell-CHSH scenario, and derive a
monogamous relation between three parties: Alice, Bob and
Charlie. We show that if Alice and Bob violate the
entropic Bell-CHSH inequality, then Alice and Charlie cannot
do so (and vice versa). An added advantage is that
 our technique also holds for
arbitrary $m$ ENC distributed over $n$ parties.

Next we turn to the experimental
realization  of our results by applying the entropic Bell-CHSH inequality to two different cases.
In the first case, all the parties share a mixed tripartite state, such that Alice and Bob
share a maximally entangled state with probability $p$, and Alice and Charlie share the same with probability $1-p$. In the second case, all the parties together share an entangled tripartite pure state. In both the cases, we observe that the monogamous relation is satisfied at all
times, while only one of the inequalities is capable of exhibiting a violation. It should be noted that experiments on an NMR quantum information processor only yield an expectation value of the observable, while probabilities of various outcomes cannot be addressed directly. This makes the task of evaluating entropic quantities on this system
difficult. To circumvent this difficulty, we provide a novel recipe to evaluate the entropy of an observable on an NMR quantum information processor.  In the experiment, all the pure and mixed states were prepared with an experimental fidelity
greater than $0.92$ and $0.96$ respectively, and the experimental results match well with
the theoretical predictions.

The paper is organized as follows: in
Sec.~\ref{ineq:sec} we begin by giving a brief description of 
entropic inequalities. In Sec.~\ref{sec:mono} we describe
our main theoretical result concerning monogamy of ENCs. In Sec.~\ref{subsec:mixed} we
provide an experimental demonstration using a mixed state, while in Sec.~\ref{subsec:pure} we
use a pure tripartite state for the experiment. Sec.~\ref{concl} 
contains some concluding remarks.

\section{Entropic inequalities and their monogamy}
\label{theory-results}
In this section 
we describe
our theoretical results, which will be used later in the paper.
We begin by providing
a brief review of ENC inequalities and focus
particularly on the entropic version of the Bell-CHSH inequality.  

\subsection{The entropic inequalities}
\label{ineq:sec}
An experiment corresponding to a contextuality inequality can be represented by a graph ~\cite{cabello-prl-14}. Mathematically, a graph $G$ is defined by a set of vertices $V$ and and a set of edges $E$, such that $G = (V, E)$.  In relation to the contextuality inequality, the set $V$ can represent either events, projectors or observables depending upon the scenario, while the set $E$ represents a relationship between different elements of $V$, such as orthogonality, exclusivity or commutativity~\cite{cabello-prl-14}.  In this paper we use commutativity graphs to study ENC inequalities, in which the set of vertices corresponds to observables and two vertices are connected by an edge if they commute with each other~\cite{kurzy-prl-12}.

Consider an $n$-cycle commutation graph in which the $n$
vertices represent observables $X_i$ and the existence of an edge
indicates that the corresponding observables commute.
An example of one such graph for $n = 4$ is given in Fig.~\ref{fig:entropy} (top panel). 

We assume that a
non-contextual joint probability distribution exists over
the entire set of observables considered, even though most
of them do not commute. It is our aim to construct a
condition based on the preceding assumption, for which a
violation would indicate that such a non-contextual joint
probability distribution does not exist.

The existence of a non-contextual joint probability
distribution over the observables $X_i$ implies that it is
possible to define a joint Shannon entropy
$H(X_0,...,X_{n-1})$ of them. We can then write
\begin{equation} 
H(X_0,X_{n-1})\leq H(X_0,...,X_{n-1}),
\label{eq:ent1} 
\end{equation}
 where the relationship
$H(X)\leq H(X,Y)$ is physically motivated by the fact that
two random variables cannot contain less information than a
single one of them. Furthermore, with repeated application
of the chain rule $H(X,Y)=H(X|Y)+H(Y)$, where $H(X|Y)$
denotes the conditional entropy of observable $X$ given
information about observable $Y$, the right hand side of the
inequality can be re-written as, \begin{widetext}
\begin{equation} \begin{aligned} H(X_0,...,X_{n-1})&\leq
H(X_0|X_1,...,X_{n-1})+H(X_1|X_2,...X_{n-1})+...+H(X_{n-2}|X_{n-1})+H(X_{n-1})\\
&\leq
H(X_0|X_1)+H(X_1|X_2)+...+H(X_{n-2}|X_{n-1})+H(X_{n-1}).
\label{eq:ent2} \end{aligned} \end{equation} \end{widetext}
The latter inequality in Eq.~\eqref{eq:ent2} is a
consequence of the relationship $H(X|Y)\leq H(X)$, which
implies that conditioning cannot increase the information
content of a random variable.  Plugging Eq.~\eqref{eq:ent2}
in Eq.~\eqref{eq:ent1} and using
$H(X_0|X_{n-1})=H(X_0,X_{n-1})-H(X_{n-1})$, we finally get
the required entropic non-contextuality inequality,
\begin{equation} H_{K_1}: H(X_0|X_{n-1})\leq
H(X_0|X_1)+...H(X_{n-2}|X_{n-1}).  \label{eq:ent_main}
\end{equation} A violation of Eq.~\eqref{eq:ent_main} would
then indicate that a non-contextual joint probability
distribution over the corresponding set of observables does
not exist.
	
It should be noted that no assumptions have been made
regarding the nature of the observables $X_i$. For all
intents and purposes they can correspond to projective
measurements or POVMs with any number of outcomes.
Furthermore, the observables could correspond to local
scenarios or non-local, in which case the entropic
inequality would be termed as non-contextual or Bell
non-local.  Since we consider generalized scenarios, we term
all entropic inequalities as non-contextual as it subsumes
the non-local scenarios as well.

One of the many interesting cases arises for $n = 4$
observables, which corresponds to the well known Bell-CHSH
scenario.  Consider two parties, Alice and Bob, each having
two observables labelled $\lbrace A_0, A_1\rbrace$ and
$\lbrace B_0, B_1\rbrace$ respectively, such that
observables of any one party do not commute, while
observables of different parties commute. Differing from the
traditional Bell-CHSH scenario, it is also assumed that a
measurement of the observables can have an arbitrary number of
outcomes. The corresponding 
entropic inequality is written as, 
\begin{equation}
H_{K_1}:H(A_1|B_1)-H(A_1|B_0)-H(B_0|A_0)-H(A_0|B_1)\leq 0.
\label{eq:entropic_chsh} 
\end{equation}

A violation of the above inequality implies a non-existence
of a joint probability distribution over all the observables
$A_x$ and $B_y$ $\forall x, y \in \lbrace 0, 1\rbrace$. 
In the next section, we show that ENC inequalities admit a
monogamous relationship which can be derived using a graph
theoretic formalism. We particularly focus on the entropic
Bell-CHSH scenario described above and show that it admits a
monogamous relationship.
\begin{figure}
\includegraphics[scale=0.8]{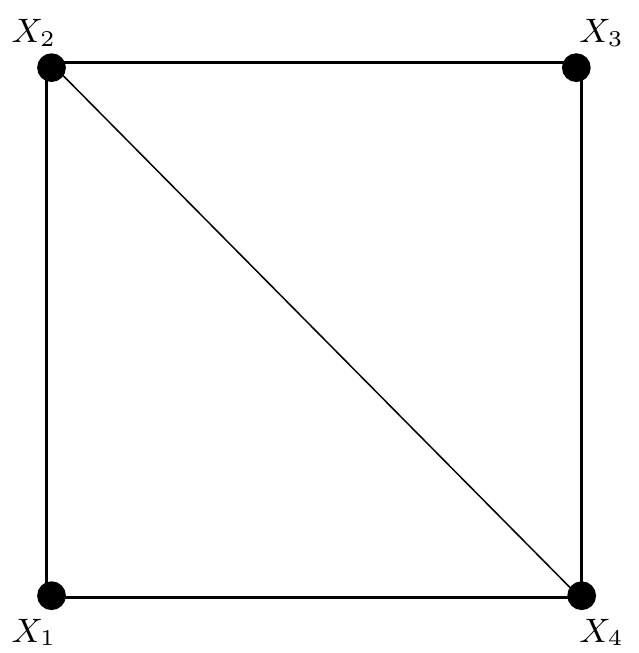}
\caption{The chordal subgraph where the vertices represents observables $X_i$ with $i=1, 2, 3, 4$ and the edge indicate commutativity relationship between the observables.}
\label{fig:generalized}
\end{figure}

\subsection{Monogamy of entropic inequalities}
\label{sec:mono}
We now elucidate the formalism to 
check and derive a monogamous relationship for any arbitrary scenario using the
graph theoretic formalism.
It should be remembered that a violation of the inequality~\eqref{eq:entropic_chsh}
implies that a joint probability distribution 
cannot exist over the given observables.

\begin{prop}
A monogamous relationship for a set of $n$ observables $X_i$,
corresponding to $m$ non-contextuality scenarios exists if
their joint commutation graph can be vertex decomposed into
$m$ chordal subgraphs such that all edges appearing in the
original $m$ non-contextuality graphs must appear just
once in the decomposition. The monogamous relationship then reads as,
\begin{equation}
H_{K_1}+H_{K_2}+H_{K_3}+....+H_{K_m} \leq 0,
\label{eq:std_mono1}
\end{equation}
where $H_{K_i}$ denotes the $i$th ENC inequality.
\label{prop:prop2}
\end{prop}
\begin{proof}

A chordal graph is a graph in which all cycles of four or more vertices have a chord. A chord is an edge that is not part of the cycle but which connects two of the vertices.
It should be noted that the definition of chordal subgraphs
implies that these subgraphs must have an edge connecting
two vertices, such that any induced cycles in the
subgraph have a length equal to three. Any induced cycles of length greater than $3$
should then have an edge such that the resultant induced cycle satisfies the 
definition.

For these $3$-cycle induced subgraphs, it is possible
to write down the corresponding entropic inequality.  Since
it is a $3$-cycle graph, a joint probability distribution
over the observables always exists~\cite{ramanathan-prl-12}
and the entropic inequality is never violated. Therefore, we
can add all the entropic inequalities obtained in this
fashion to obtain
\begin{equation}
\sum_l\sum_{i, j = 1}^{n}g_l(H'(X_i|X_j)),
\label{eq:grouping}
\end{equation}
where $H'(X_i|X_j) = H(X_i|X_j)$ only if the vertices $i$ and $j$ belong to an edge and zero 
otherwise, $g(X) = \pm X$ depending on the $3$ cycle chosen and $l$ is the number 
of edges appearing in the $m$ chordal graphs. It should be noted that Eq.~\eqref{eq:grouping}
is a linear combination of $H(X_i|X_j)$ and since we require that all the terms appearing in the individual $m$
ENC inequalities appear in the decomposition, it is
possible to obtain the form in Eq.~\eqref{eq:std_mono1} via
linear manipulation of the ENC inequalities of the $3$-cycle
induced graphs. It should be noted that each term appearing in
Eq.~\eqref{eq:ent_main} corresponds to an edge in the
commutativity graph and missing even a single edge in the
chordal decomposition, would make one of the
non-contextuality inequalities incomplete and
Eq.~\eqref{eq:std_mono1} unachievable. 
Furthermore, the linear manipulations required correspond to choosing a suitable form of 
$g(X)$ such that the $m$ ENC inequalities can be obtained by 
grouping certain terms together. This manipulation depends on the commutation graph of the 
interested scenario. 

As an example of our technique consider the chordal graph given in Fig.~\ref{fig:generalized} which 
is formed by two $3$-cycle graphs.  We assume that this is one of the subgraphs obtained via vertex decomposition of a joint commutativity graph such that the term $H(X_2|X_4)$ corresponding to the edge $(X_2, X_4)$ does not appear in Eq.~\eqref{eq:std_mono1}, while the terms corresponding to the other edges do appear. In order to eliminate this term, we consider the ENC inequality for the $3$-cycle graph with vertices $(X_1, X_2, X_4)$, written in a cyclic form as,
\begin{equation}
H(X_1|X_4)-H(X_1|X_2)-H(X_2|X_4) \leq 0,
\label{eq:ch_ind_1}
\end{equation}
while for the other $3$-cycle graph formed by the vertices $(X_2, X_3, X_4)$ we use an anti-cyclic form of the ENC inequality as,
\begin{equation}
H(X_2|X_4)- H(X_2|X_3)-H(X_3|X_4) \leq 0.
\label{eq:ch_ind_2}
\end{equation} 

Since all $3$-cycle graphs admit a joint probability distribution, the aforementioned inequalities are always satisfied and can therefore be added to give,
\begin{equation}
H(X_1|X_4)- H(X_1|X_2)-H(X_2|X_3)- H(X_3|X_4) \leq 0,
\label{eq:ch_ind_3}
\end{equation} 
which is the ENC inequality over the graph in Fig~\ref{fig:generalized} without the edge $(X_2, X_4)$. We use this technique to derive our monogamous relationships. In the case when such a vertex 
decomposition cannot be found, we cannot guarantee that a monogamous relation will exist.

\begin{figure}
\includegraphics[scale=1]{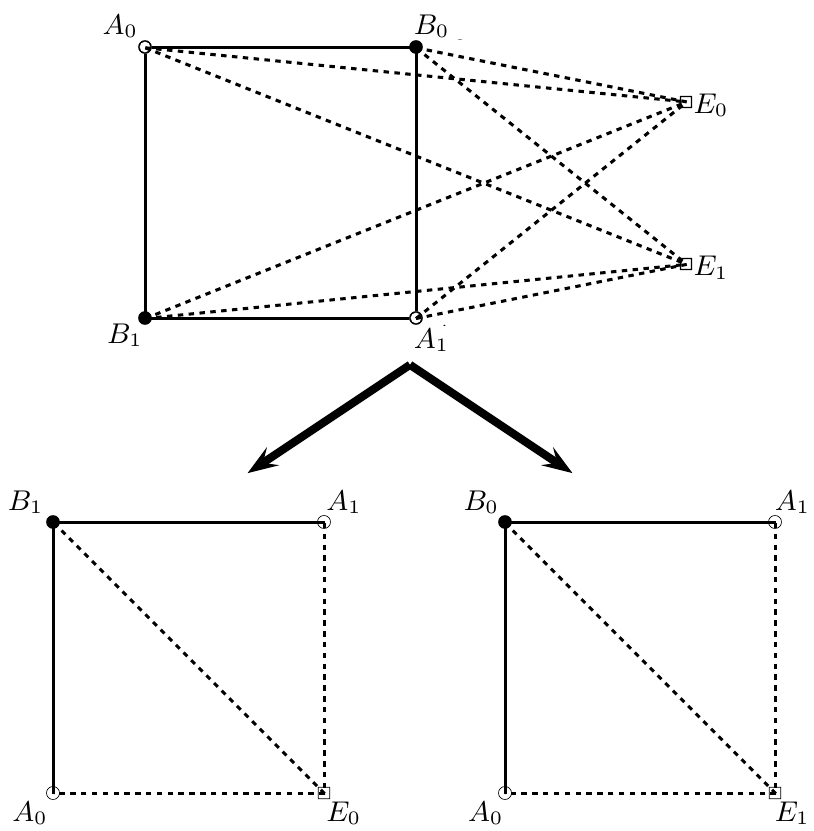}
\caption{The joint commutation graph of Alice-Bob-Eve (top)
and its chordal decomposition (below) 
according to Proposition~\ref{prop:prop2}, where solid lines
indicate commutativity between observables 
of Alice and Bob, while dashed lines indicate commutativity
between Alice (Bob) and Charlie.}
\label{fig:entropy}
\end{figure}
	
We now apply this technique to derive a monogamous relationship of two entropic CHSH inequalities. Consider a standard monogamous relationship 
between two entropic non-contextuality scenarios, $K_1$ and
$K_2$ which should dictate,
\begin{equation}
H_{K_1}+H_{K_2} \leq 0,
\label{eq:std_mono}
\end{equation}
where $H_{K_{1(2)}}\leq 0$ denotes the existence of a joint
probability distribution over the $n$ 
observables. Since the monogamy relationship must be 
obeyed at all times, Eq.~\eqref{eq:std_mono} then physically implies that 
a joint 
probability distribution over all the observables appearing in 
the new entropic inequality, $H = H_{K_1}+H_{K_2}$ exists.

As explained, in order to achieve a form of
Eq.~\eqref{eq:std_mono}, the joint graph of the
non-contextuality scenarios must be decomposable into
chordal subgraphs, which admit a joint probability
distribution. However, we also require an additional feature
that all edges of the individual non-contextuality graphs
must appear at least once in any of the chordal subgraphs.

Consider the CHSH scenario in which two parties Alice and
Bob can perform a measurement of the observables $A_0$,
$A_1$ and $B_0$, $B_1$ respectively. We assume a third party
Charlie with observables $E_0$ and $E_1$, which commute with
the observables of Alice and Bob. The scenario is
illustrated by a joint graph as shown in
Fig.~\ref{fig:entropy}, where vertices represent observables
and edges indicate the commutation relationship. Without
loss of generality we assume that Charlie would like to
violate the entropic CHSH inequality with Alice. The two
corresponding Alice-Bob and Alice-Charlie entropic CHSH
inequalities are,

\begin{equation}
\begin{aligned}
H_{K_1}: &H(A_0|B_0)-H(A_0|B_1)\\
&-H(B_1|A_1)-H(A_1|B_0)\leq 0,
\end{aligned}
\end{equation}

\begin{equation}
\begin{aligned}
H_{K_2}: &H(A_0|E_0)-H(A_0|E_1)\\
&-H(E_1|A_1)-H(A_1|E_0)\leq 0.
\end{aligned}
\end{equation}
The joint commutation graph is decomposed into two chordal
graphs while keeping the edges appearing in the individual
Alice-Bob and Alice-Charlie CHSH scenarios intact in the
decomposition, as shown in Fig.~(\ref{fig:entropy}). The
corresponding ENC inequalities for each $3$-cycle graph in
every chordal subgraph are given as, 
\begin{eqnarray}
H(A_0|E_0)-H(A_0|B_1)-H(B_1|E_0)\leq 0,\\
H(B_1|E_0)-H(B_1|A_1)-H(A_1|E_0)\leq 0,\\
H(A_0|B_0)-H(A_0|E_1)-H(E_1|B_0)\leq 0,\\
H(E_1|B_0)-H(E_1|A_1)-H(A_1|B_0)\leq 0,
\end{eqnarray}

Being cyclic and chordal, all the above inequalities are a
necessary and sufficient condition for a joint probability
distribution to exist. Furthermore, we have carefully chosen
the terms with positive and negative coefficients so as to
achieve a final form according to Eq.~\eqref{eq:std_mono}.
Adding the above inequalities and grouping the terms
according to $H_{K_1}$ and $H_{K_2}$, we obtain
\begin{equation}
H_{K_1}+H_{K_2}\leq 0,
\label{eq:mono_ent_chsh}
\end{equation}
which is the required monogamy relationship. We note that
this monogamy relationship was also derived
in~\cite{entropic_no_signalling}, albeit in a different
manner and specifically for the entropic CHSH inequality.
However, our formalism can be readily generalized to $n$
observables distributed among $m$ parties.  \end{proof} The
derived monogamy relationship~\eqref{eq:mono_ent_chsh}
imposes severe restrictions on the violation of
Alice-Charlie entropic CHSH inequality. 

\section{Experimental demonstration}
\label{sec:monogamy:exp}

In this section, we experimentally demonstrate the monogamy relationship 
derived for the entropic Bell-CHSH scenario~\eqref{eq:mono_ent_chsh} on an 
NMR quantum information processor, using two different set of states. We show that 
for both the sets of states, the entropic Bell-CHSH obeys the monogamy relationship 
we derived above.

\subsection{Implementation using a mixed tripartite state}
\label{subsec:mixed}

We experimentally implement the 
monogamy inequality~\eqref{eq:mono_ent_chsh} using a mixed 
tripartite state which is a classical mixture of two pure maximally entangled states given as: 

\begin{equation}
\rho = p (\vert \psi_1 \rangle \langle \psi_1 \vert) + (1-p) (\vert \psi_2 \rangle \langle \psi_2 \vert),
\label{eq:mono_mixed}
\end{equation}
where 

\begin{equation}
\begin{aligned}
\vert \psi_1 \rangle &= \frac{1}{\sqrt{2}}(\vert 001 \rangle + \vert 111 \rangle),  \\
\vert \psi_2 \rangle &= \frac{1}{\sqrt{2}}(\vert 010 \rangle + \vert 111 \rangle),  
\end{aligned}
\end{equation}
and $p\in\left[0, 1\right]$. As can be seen, the state $\rho$ physically implies that Alice and Bob share a 
maximally entangled state with probability $p$ while Charlie is separable, and with probability $1 - p$, Alice and  Charlie share a maximally entangled state while Bob is separable.

The observables of Alice, Bob and Charlie are assumed to lie in the $X-Z$ plane and correspond to Pauli spin measurements along the unit vectors $\bm{a}, \bm{a'}, \bm{b}, \bm{b'}$ and $\bm{e}, \bm{e'}$ respectively. The vectors $\bm{a}, \bm{b'}, \bm{a'}$ and $\bm{b}$ are successively separated by an angle $\frac{\theta}{3}$, while the vectors $\bm{e}$ and $\bm{e'}$ are taken to be the same as Bob's. The corresponding Pauli observables are then given as: 
\begin{equation} \label{eq1}
\begin{aligned}
 &A_0 = 
\begin{bmatrix}
1 & 0 \\
0 & -1 
\end{bmatrix}, \quad 
&&A_1 = 
 \begin{bmatrix}
\cos\frac{2\theta}{3} & -\sin\frac{2\theta}{3} \\
\sin\frac{2\theta}{3} & \cos\frac{2\theta}{3} 
\end{bmatrix} \\
&B_0 =
 \begin{bmatrix}
\cos\theta & -\sin\theta \\
\sin\theta & \cos\theta
\end{bmatrix}, \quad 
&&B_1 =
 \begin{bmatrix}
\cos\frac{\theta}{3} & -\sin\frac{\theta}{3} \\
\sin\frac{\theta}{3} & \cos\frac{\theta}{3} 
\end{bmatrix},
\end{aligned}
 \end{equation}
while the observables of Charlie ($E_0, E_1$) are the same as Bob's but acting in a different Hilbert space. All of the aforementioned observables have eigenvalues $a_i, b_i, e_i \in \{-1,+1\}$ and follow the commutativity conditions as shown in Fig.~\ref{fig:entropy}.

For the set of observables given in Eq.~\eqref{eq1}, the maximum violation of the Alice-Bob entropic CHSH inequality $H_{K1}=0.237$ bits is found to be at $\theta = 0.457$ radians when 
the parties share the state $\ket{\psi_1}$. The same also holds for Alice-Charlie entropic CHSH inequality when the parties share the 
state $\ket{\psi_2}$.

We experimentally implemented the monogamy relation given in
Eq.~\eqref{eq:mono_ent_chsh} on an 
eight-dimensional quantum system.  We used a
molecule of ${}^{13}$C -labeled diethyl fluoromalonate dissolved in acetone-D6,
with the ${}^{1}$H, ${}^{19}$F and ${}^{13}$C spin-1/2 nuclei being encoded as
`qubit one', `qubit two' and `qubit three', respectively (see Fig.~\ref{mole
and pps }(a)). The system was initialized in the pseudopure state (PPS)
$|000\rangle$  using the spatial averaging technique~\cite{mitra-jmr-07} and the
corresponding NMR spectra and experimental tomographs are given in
Fig.~\ref{mole and pps }(a),(b).

Experiments were performed
on a Bruker Avance III 600-MHz FT-NMR
spectrometer equipped with a QXI probe.  Local unitary operations 
were achieved
by RF pulses of suitable amplitude, phase, and duration and nonlocal
unitary operations were
achieved by free evolution under the system Hamiltonian.
The $T_1$ and $T_2$ 
relaxation times
of ${}^{1}$H, ${}^{19}$F,  ${}^{13}$C spin-1/2
nuclei range from 4.16 sec to 7.16 sec and 0.99 sec to 3.56 sec, respectively.
The duration of the $\frac{\pi}{2}$ pulses for ${}^{1}$H, ${}^{19}$F, and
${}^{13}$C nuclei are 9.5 $\mu$s at 18.14 W power level, 22.2 $\mu$s at a power
level of 42.27 W, and 15.2 $\mu$s at a power level of 179.47 W, respectively.
At room temperature, NMR experiments are only sensitive to the deviation
density matrix and the initial state is prepared from the thermal equilibrium
into a PPS:

\begin{equation}
\rho_{000} = \frac{1 -	\epsilon}{2^3} I_8  + \epsilon |000\rangle \langle 000|
\label{eq:pps daviasion}
\end{equation}

where $\epsilon \approx 10^{-6}$ and $I_8$ is $8 \times 8 $ identity operator.
The system is initialized in the PPS
state \ie $|000\rangle$ using the spatial averaging technique
\cite{mitra-jmr-07}, which is based on dividing the system in sub-ensembles
which can be accessed independently in NMR
by using a combination of RF pulses and pulsed magnetic gradients.  The state
tomography was 
performed using the least square optimization
technique\cite{gaikwad-qif-21} with an experimental state fidelity of 0.98.

\begin{figure}
\centering
\includegraphics[scale=1.0]{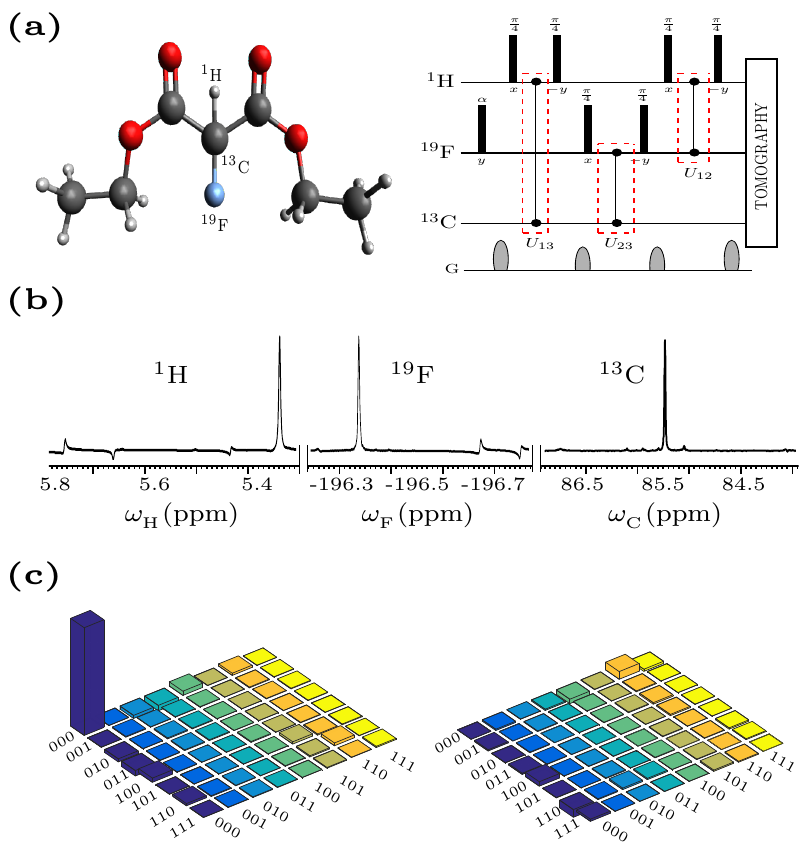}
\caption{(a) Molecular structure of $\rm^{13}C$ - labeled diethyl
fluoromalonate. $\rm^{1}H$, $\rm^{19}F$, $\rm^{13}C$ are used as three qubits
in the experiment. NMR pulse 
sequence for the initialization of the state into
the PPS $|000\rangle$ state. The value of the flip angle $\alpha$ is kept at
$57.87^{\circ}$. 
The phases and the flip angles are written with each pulse.
The $U_{12}, U_{23}, U_{13}$ are evolution operators with time intervals
$\tau_{12}$, $\tau_{23}$, $\tau_{13}$ equal to
$\frac{1}{2J_{HF}}$,
$\frac{1}{2J_{HC}}$, and $\frac{1}{2J_{FC}}$
respectively, shown by red dotted lines. (b) The
NMR spectra corresponding to the PPS $|000 \rangle$. The
horizontal scale represents the chemical shifts in ppm. (c)  Real (left) and
imaginary (right) parts of the experimentally reconstructed tomographs of the
state $ |000 \rangle$ with an experimental state fidelity of
0.98.}

\label{mole and pps }
\end{figure}

We began by preparing the mixed tripartite state given in
Eq.~\eqref{eq:mono_mixed} for different
values of $p$. In order to achieve this,
we utilized the
temporal averaging technique~\cite{oliveira-book-07,
knill-pra-98}. Using this technique, it is possible to prepare arbitrary mixed
states on an NMR quantum information processor, by applying suitable unitary
transformations on some common initial PPS. The different experiments 
are
performed on common initial states, the results of which are independently
stored. Finally, these results are combined to produce an average state which
simulates the behavior of a mixed state.

In our case, we prepared the mixed state given in  Eq.~\eqref{eq:mono_mixed},
which is a mixture 
of two pure states $\vert \psi_1 \rangle$, $\vert \psi_2
\rangle$. We first prepared these two states  by applying suitable unitaries on
the initial state $\vert 000 \rangle \langle 000\vert$ in two different and
independent experiments~\cite{singh-pla-19}.  The states of these two
experiments are then added with appropriate probabilities to achieve the
desired mixed state given in Eq.~\eqref{eq:mono_mixed}. To demonstrate the
monogamy relation, we chose different values of $p$ to experimentally prepare
the desired mixed state with  experimental state fidelities $\geq 0.956$. The
tomograph of one such experimentally prepared state with  $p=1$ is shown in 
Fig.~\ref{tomomixed}.

\begin{figure}
\centering
\includegraphics[scale=1.0]{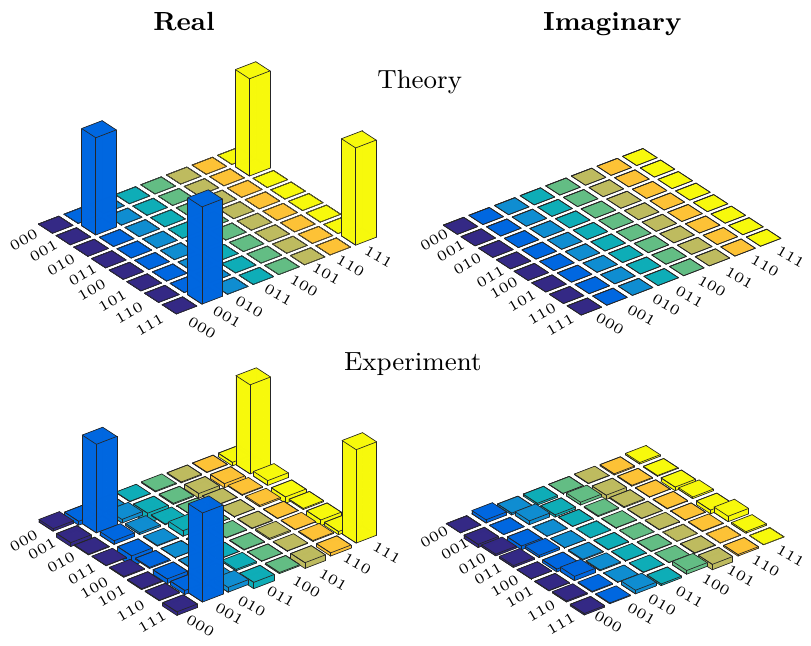}
\caption{Real and imaginary parts of the
tomographs of the theoretically and
experimentally reconstructed density operator for the tripartite state $\rho$
with $p=1$ and an experimental state fidelity of $0.97$.}
\label{tomomixed}
\end{figure}

After experimentally preparing the states, we measured the desired
probabilities in order to calculate the entropies involved in the
inequality given in Eq.~\eqref{eq:mono_ent_chsh}. In order to 
calculate the probabilities, we
transformed the required probabilities in terms of expectation values.  This is
necessary because experiments on an NMR quantum information processor yield
only expectation values of the observables. These expectation values are
evaluated by decomposing the observables in terms of linear combinations of
Pauli operators which can be mapped to the single-qubit Pauli $Z$ operator.
This mapping is particularly useful in the context of an NMR experimental setup
where the expectation value of the $Z$ operator is easily accessible and
corresponds to the observed $z$ magnetization of a nuclear spin in a particular
quantum state. The normalized experimental intensities of the NMR signal then
provide an estimate of the expectation value of the Pauli $Z$ operator in that
quantum state~\cite{dileep-pra-19}.

The probabilities can be written as $P(A_i=a_i, B_j=b_j)=\text{tr}(\rho\vert
a_{i} \rangle \langle a_{i} \vert \otimes |b_{j}\rangle \langle b_{j} \vert
\otimes I )$ and $P(A_i=a_i, E_j=e_j)=\text{tr}(\rho\vert a_{i} \rangle \langle
a_{i} \vert \otimes I  \otimes |e_{j}\rangle \langle e_{j} \vert )$ where
$\vert a_{i} \rangle$, $\vert b_{j}\rangle$, $\vert e_{j}\rangle$ are the
eigenvector of the observables corresponding to Alice, Bob and Charlie,
respectively. 
The observables ($\vert a_{i}\rangle \langle a_{i} \vert
\otimes |b_{j}\rangle \langle b_{j} \vert \otimes I $) , ($\vert a_{i}\rangle
\langle a_{i} \vert \otimes I \otimes |e_{i}\rangle \langle e_{i} \vert)$ are
decomposed in terms of linear combinations of Pauli operators and details are
given in Appendix-A.  The idea is to unitarily map the state $\rho$ to another
state $\rho'$, such that $\langle X \rangle _{\rho} = \langle I_{iz}\rangle
_{\rho'}$ where $X$ is the observable to be measured in the state $\rho$ and
$I_{iz}$ is the z-spin angular momentum of the qubit. This can be achieved by
measuring the $I_{iz}$ on the state $\rho'$.  For example, one can find the
expectation values of $ \langle \sigma_x  \otimes \sigma_x \otimes I \rangle
_\rho $ and $\langle \sigma_x \otimes I \otimes \sigma_x \rangle_\rho  $ which
are involved in the evaluation of probabilities 
(see Appendix-A), by using the
quantum circuit and corresponding NMR pulse sequence given in Fig.\ref{map
state}(a) and Fig.\ref{map state}(b) respectively, where the implementation is
followed by a measurement of the spin magnetization of 
the second and third
qubits, respectively.

\begin{figure}
\centering
\includegraphics[scale=1.0]{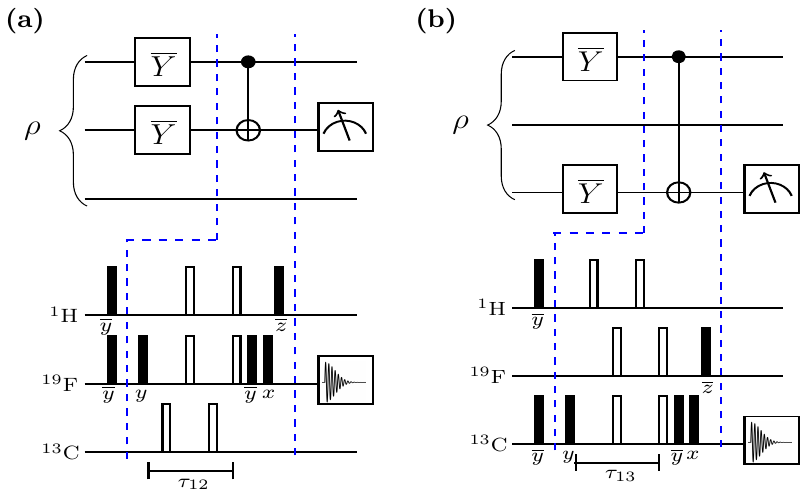}
\caption{Quantum circuit and corresponding NMR pulse sequence  to map the state $\rho$ to the state $\rho'$ such that (a) $ \langle \sigma_x  \otimes \sigma_x \otimes I \rangle _\rho = \langle I_{2z} \rangle _{\rho'} $, (b) $ \langle \sigma_x \otimes I \otimes \sigma_x \rangle_\rho  = \langle I_{3z} \rangle _{\rho'} $  . In the NMR pulse sequences, the filled rectangles are $\frac{\pi}{2}$ RF pulses while unfilled rectangles denotes the $\pi$ RF pulses. The phase of each pulse is written over the respective pulse and bar over the phase denotes the negative phase. The free evolution of time is denoted by $\tau_{12}, \tau_{13}$. }
\label{map state}
\end{figure}

We experimentally calculated the ENC inequalities $H_{K1}$, $H_{K2}$ and the
ENC monogamy relation $H_{K1}+H_{K2}$ for the tripartite mixed state given in
Eq.~\eqref{eq:mono_mixed} for different values of $p$. Experimental values of
$H_{K1}$, $H_{K2}$, $H_{K1}+H_{K2}$ with respect to various values of $p$ are
plotted in Fig.~\ref{mixed state graph}.  It can be seen that $H_{K1}$ is
violated for the tripartite state with $p=1$, while $H_{K2}$ is violated for
the state with $p=0$. The monogamy relation $H_{K1}+H_{K2}$ is never violated
for any value of $p$ and the results are in good agreement with the theoretical
predictions.

\begin{figure}
\centering
\includegraphics[scale=1.0]{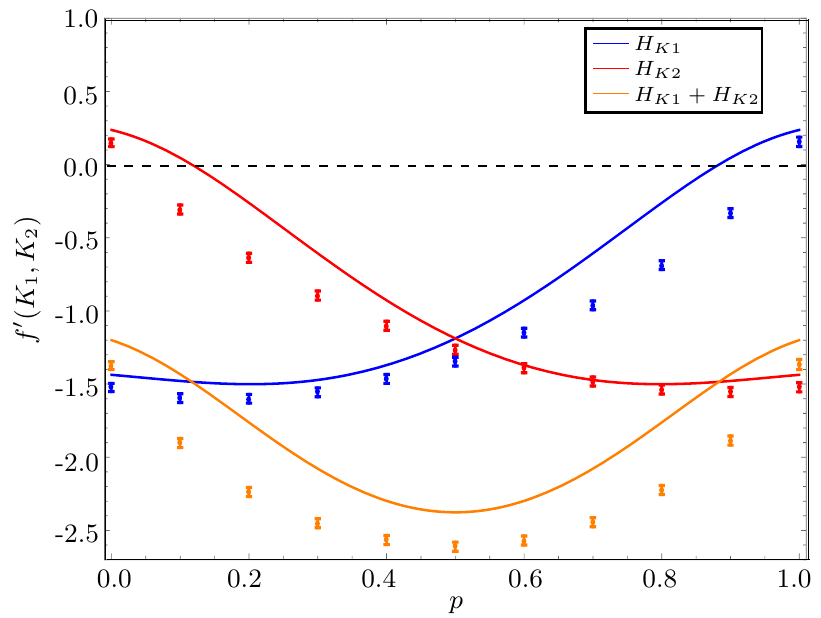}
\caption{The experimentally calculated values for the monogamy inequality
Eq.~\eqref{eq:mono_ent_chsh}  denoted by $f'(K_1, K_2)$ for the tripartite mixed states
given in Eq.~\eqref{eq:mono_mixed} for different values of $p$. The black dotted line
is the maximum non-contextual value 
that can be achieved by inequalities $H_{K1}$, $H_{K2}$
and  $H_{K1}+H_{K2}$ to follow the ENC theories and monogamy relationship,
respectively. The red color line and the red points with error bars 
represent
the theoretical and experimental calculated values for the 
inequality $H_{K1}$, respectively. 
Similar representations for the inequalities $H_{K2}$ and
$H_{K1}+H_{K2}$ are used.}
\label{mixed state graph}
\end{figure}

It is seen that the experimental values are always lower than the corresponding
theoretical values. This is due to the fact that the maximum value of the
inequality is always achieved for a pure state and any addition of noise makes
it a mixed state. Therefore, the value is always observed to be lower than the
theoretical one. The difference between them also increases with increasing
values of $p$ upto $p = 0.5$.  For this region, the state prepared for $p =
0.0$ is the dominant one which is shown to violate $H_{K_1}$. Furthermore, the
corresponding inequalities are logarithmic in nature, which also compounds the
errors for states till $p = 0.5$. However, after this value the state prepared
for $p = 1.0$  dominates and the errors again follow a similar trend. It should
be noted that the experimental curve follows the same trend as the
theoretical curve.

\subsection{Implementation using a pure tripartite state}
\label{subsec:pure}
In this subsection, we experimentally test the monogamy relation
Eq.~\eqref{eq:mono_ent_chsh} for a pure tripartite state, having two parameters
which we can vary. We show that the monogamy relation holds and is in
good agreement with theoretical predictions.

We take a pure tripartite state of the form,
\begin{equation}
\vert \phi \rangle = N (p_1 \vert 001 \rangle +  p_2 \vert 010 \rangle + (p_1+p_2) \vert 111 \rangle),
\label{eq:mono_tripartite pure state}
\end{equation}
where $N = \frac{1}{\sqrt{p_1^2 + p_2^2 + (p_1 + p_2)^2}} $ is the
normalization factor. We experimentally prepared five different states
corresponding to various values of $p_1$ and $p_2$.  The quantum circuit and
the corresponding NMR pulse sequence is given in Fig.~\ref{ckt pure state}(a),
(b). Different pure states corresponding to various values of $p_1$ and
$p_2$ were generated by suitably choosing the values of $\theta_1, \theta_2$
and $\theta_3$. Tomograph of one such experimentally prepared state with
$p_1=0.25$ and $p_2=0.50$ is depicted in Fig.~\ref{tomo pure}, with an
experimental state fidelity of $0.93$. After preparing the states, we measured
the desired probabilities by mapping the state onto the Pauli basis operators
in order to calculate the entropies involved in the inequality
Eq.~\eqref{eq:mono_ent_chsh} as discussed earlier in Sec.~\ref{subsec:mixed}.

\begin{figure}
\centering
\includegraphics[scale=1.0]{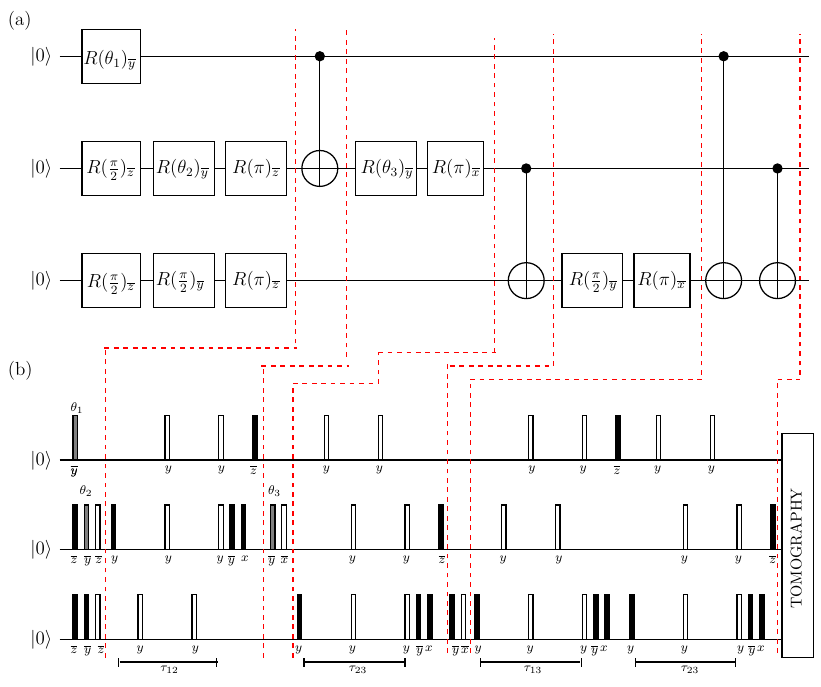}
\caption{(a) Quantum circuit for 
state preparation for the tripartite pure state
$\vert \phi \rangle$ and (b) corresponding NMR pulse sequence. 
The filled rectangles 
denote $\frac{\pi}{2}$ RF pulses while the unfilled
rectangles denote $\pi$ RF pulses. 
The phase of each pulse is written over
the respective pulse, and bar over the phase denotes a 
negative phase.
$\tau_{12}, \tau_{23}, \tau_{13}$ are the free time evolutions. All the CNOTS
are separated by red dotted lines.}
\label{ckt pure state}
\end{figure}

\begin{figure}
\centering
\includegraphics[scale=1.0]{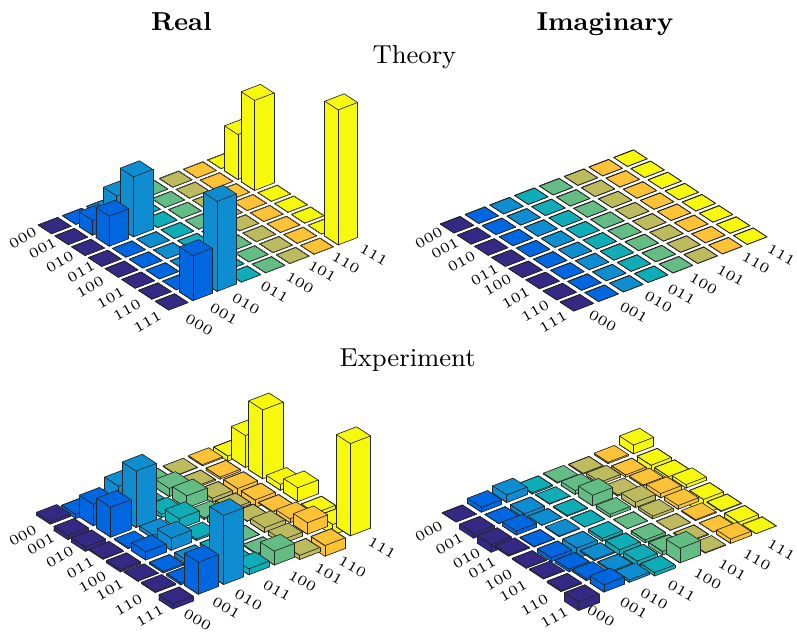}
\caption{Real and imaginary part of tomograph of the theoretically and
experimentally reconstructed density operator for the tripartite pure state
$\rho'=\vert \phi \rangle \langle \phi \vert $ with $p_1=0.25, p_2=0.50 $ with
an experimental state fidelity of $0.93$.}
\label{tomo pure}
\end{figure}

We experimentally evaluate $H_{K_1}$, $H_{K_2}$ and $H_{K_1}+H_{K_2}$
with respect to the different values of $p_1$ and $p_2$ and the results
are given in 
Table-\ref{inequality values pure state}. We see that 
the inequality $H_{K_1}$ is violated for $p_1 = 1$ and $p_2 = 0$ and  $H_{K_2}$ is violated for $p_1 = 0$ and $p_2 = 1$, while the 
monogamy relationship $H_{K_1}+H_{K_2}$ is never violated for any 
value of $p_1$ and $p_2$. This shows that the monogamy relation between the ENC inequalities is obeyed for all such possible states.

\begin{table*}
\setlength{\tabcolsep}{8pt} 
\renewcommand{\arraystretch}{1.5}
\caption{Theoretically computed and experimentally measured values of the
monogamy inequalities for a pure state.}
\vspace*{12pt}
\centering
\begin{tabular}{|cc|cc|cc|cc|}
\hline
 &   & \multicolumn{2}{c|}{$H_{K_1}$}
 &    \multicolumn{2}{c|}{$H_{K_2}$} &  \multicolumn{2}{c|}{$H_{K_1}+H_{K_2}$}\\
\hline
$p_1$ &  $p_2$ & Theory  &
Experiment & Theory & Experiment  & Theory & Experiment \\
\hline
\hline
1.00 & 0.00 & 0.236 & 0.156$\pm$0.032 & -1.436 & -1.522$\pm$0.035 & -1.200 & -1.366$\pm$0.034 \\
0.50 & 0.25 & -0.492 & -0.606$\pm$0.021 & -1.338 & -1.413$\pm$0.027 & -1.830 & -2.019$\pm$0.024 \\
0.50 & 0.50 & -1.017 & -1.103$\pm$0.022 & -1.017 & -1.082$\pm$0.030 & -2.034 & -2.185$\pm$0.026 \\
0.25 & 0.50 & -1.338 & -1.397$\pm$0.021 & -0.492 & -0.598$\pm$0.024 & -1.830 & -1.995$\pm$0.023 \\
0.00 & 1.00 & -1.436 & -1.523$\pm$0.028 & 0.236 & 0.149$\pm$0.025 & -1.200 & -1.374$\pm$0.027 \\
\hline
\end{tabular}
\label{inequality values pure state}
\end{table*}

\section{Concluding Remarks}
\label{concl}
In this paper we develop a theoretical framework to analyze when a scenario
will exhibit a monogamous relationship. We perform our analysis of the monogamy
relations expected in a general tripartite scenario for the entropic Bell-CHSH
inequality and give the first theoretical study of the monogamy of entropic
inequalities based on the graph theoretic formalism.  We also give the first
experimental demonstration of monogamy of  entropic correlations using three
NMR qubits. Evaluating entropy on an NMR quantum processor 
is quite hard and has not been
performed earlier.  
We are able to obtain information about entropies using 
measurements of only the expectation values of observables.
We experimentally show the monogamy of entropic inequality
for a pure tripartite state as well as for the mixed state. This indicates with
certainty that our formalism holds in general. 

It should be noted that due to limitations of access to individual events,  the NMR implementation 
of such scenarios  involves use of Born rule to interpret the  results of measurements.
This is not the ideal way of carrying out  such experiments and a more refined test
would evaluate the requisite probabilities without assuming the Born rule. This is possible 
on optical systems~\cite{zhan-prl-17,qu-pra-20}, where the probabilities are calculated by estimating the frequency clicks of the photodetectors rather than the Born rule.

The experimental implementation
of monogamy inequalities is important for quantum information
processing tasks and our results are a step forward in this direction.

\begin{acknowledgments}
All the experiments were performed on a Bruker Avance-III 600 MHz FT-NMR
spectrometer at the NMR Research Facility of IISER Mohali. 
A.
acknowledges financial support from
DST/ICPS/QuST/Theme-1/2019/Q-68.  K~.D. acknowledges financial support from
DST/ICPS/QuST/Theme-2/2019/Q-74. JS acknowledges support by \href{http://dx.doi.org/10.13039/100009042}{Universidad de Sevilla} Project Qdisc (Project No.\ US-15097), with FEDER funds, \href{http://dx.doi.org/10.13039/501100011033:}{MCINN/AEI} Projet No.\ PID2020-113738GB-I00, and QuantERA grant SECRET, by \href{http://dx.doi.org/10.13039/501100011033:}{MCINN/AEI} (Project No.\ PCI2019-111885-2).
\end{acknowledgments}

\appendix
\label{appendix}
\section{Observable decomposition in terms of Pauli operators }
\label{sec:sec_app_b}
To experimentally test the inequality on an NMR quantum processor, we
have calculated the entropies involved in the inequality by 
measuring the probabilities for an experimentally prepared state. This
can be achieved by decomposing the projectors $\vert a_{\pm 1}\rangle \langle
a_{\pm 1} \vert \otimes \vert b_{\pm 1}\rangle \langle b_{\pm 1} \vert $  as
linear combinations of Pauli operators,
which can then be
mapped to a single-qubit Pauli $Z$ operator. 
In the tripartite scenario, the probabilities can be
written as $P(A_i=\pm1, B_j=\pm1)=\text{tr}(\rho\vert a_{\pm1}\rangle \langle
a_{\pm1} \vert \otimes |b_{\pm1}\rangle \langle b_{\pm1} \vert \otimes I )$ and
$P(A_i=\pm1, E_j=\pm1)=\text{tr}(\rho\vert a_{\pm1}\rangle \langle a_{\pm1}
\vert \otimes I \otimes |e_{\pm1}\rangle \langle e_{\pm1} \vert  )$ for the
Alice-Bob and Alice-Charlie entropic CHSH inequalities. To test the monogamy
inequality Eq.\eqref{eq:mono_ent_chsh}, we measure the probabilities to calculate
the entropies $H(A_i \vert B_j)$  $H(A_i \vert E_j)$ involved in the
inequalities $H_{K1}$ and $H_{K2}$.

Probabilities are calculated by decomposing the observables 
in terms of linear combination of Pauli operators as: 
 
\begin{equation}
\begin{aligned}
\text{P(A$_1=+1$)}&=0.286 \text{b$_{16}$}+0.410 \text{b$_{48}$} +0.500 I,\\
\text{P(B$_1=+1$)}&=0.476 \text{b$_{12}$}+0.149 \text{b$_{4}$}+0.500I,\\
\text{P(A$_0=+1$, B$_0=+1$)}&=0.152 \text{$b_{12}$}+0.197 \text{b$_{4}$}+0.250\text{b$_{48}$}\\
&+0.197 \text{b$_{52}$}+0.152 \text{b$_{60}$}+0.250I,\\
\text{P(A$_0=+1$, B$_0=-1$)}&=-0.152 \text{b$_{12}$}-0.197 \text{b$_{4}$}+0.250 \text{b$_{48}$} \\ &-0.197\text{b$_{52}$}-0.152\text{b$_{60}$}+0.250I,\\
\text{P(A$_0=-1$, B$_0=+1$)}&=0.152 \text{b$_{12}$}+0.197 \text{b$_{4}$}-0.250 \text{b$_{48}$}\\
&-0.197 \text{b$_{52}$}-0.152 \text{b$_{60}$}+0.250 I,\\
\text{P(A$_0=-1$, B$_0=-1$)}&=-0.152 \text{b$_{12}$}-0.197 \text{b$_{4}$}-0.250 \text{b$_{48}$}\\
&+0.197\text{b$_{52}$}+0.152 \text{b$_{60}$}+0.250\text{I},\\
\text{P(A$_0=+1$, B$_1=+1$)}&=0.238 \text{b$_{12}$}+0.074 \text{b$_{4}$}+0.250 \text{b$_{48}$}\\
&+0.074\text{b$_{52}$}+0.238\text{b$_{60}$}+0.250 I,\\
\text{P(A$_0=+1$, B$_1=-1$)}&=-0.238 \text{b$_{12}$}-0.074 \text{b$_{4}$}+0.250\text{b$_{48}$}\\&-0.074\text{b$_{52}$}-0.238 \text{b$_{60}$}+0.250 I,\\
\text{P(A$_0=-1$, B$_1=+1$)}&=0.238 \text{b$_{12}$}+0.074 \text{b$_{4}$}-0.250 \text{b$_{48}$}\\&-0.074\text{b$_{52}$}-0.238 \text{b$_{60}$}+0.250 I,\\
\text{P(A$_0=-1$, B$_1=-1$)}&=-0.238 \text{b$_{12}$}-0.074 \text{b$_{4}$}-0.250\text{b$_{48}$}\\&+0.074\text{b$_{52}$}+0.238 \text{b$_{60}$}+0.250 I,\\
\end{aligned}
\end{equation}

\begin{equation}
\begin{aligned}
\text{P(B$_1=+1$, A$_1=+1$)}&=0.205  \text{b$_{12}$}+0.074 \text{b$_{16}$}+0.042 \text{b$_{20}$}\\&
+0.061 \text{b$_{28}$}+0.143 \text{b$_{4}$}+0.238 \text{b$_{48}$} \\
&+0.136 \text{b$_{52}$}+0.195 \text{b$_{60}$}+0.250 I,\\
\text{P(B$_1=+1$, A$_1=-1$)}&=-0.205\text{b$_{12}$}+0.074 \text{b$_{16}$}-0.042 \text{b$_{20}$}\\&-0.061\text{b$_{28}$}-0.143 \text{b$_{4}$}+0.238\text{b$_{48}$} \\
&-0.136\text{b$_{52}$} -0.195 \text{b$_{60}$}+0.250 I,\\
\text{P(B$_1=-1$, A$_1=+1$)}&=0.205\text{b$_{12}$}-0.074 \text{b$_{16}$}-0.042 \text{b$_{20}$}\\&-0.061 \text{b$_{28}$}+0.143\text{b$_{4}$}-0.238\text{b$_{48}$}\\
&-0.136 \text{b$_{52}$}-0.195\text{b$_{60}$}+0.250I,\\
\text{P(B$_1=-1$, A$_1=-1$)}&=-0.205 \text{b$_{12}$}-0.074 \text{b$_{16}$}+0.042 \text{b$_{20}$}\\&+0.061 \text{b$_{28}$}-0.143 \text{b$_{4}$}-0.238 \text{b$_{48}$}\\
&+0.136 \text{b$_{52}$}+0.195 \text{b$_{60}$}+0.250 I,\\
\end{aligned}
\end{equation}

\begin{equation}
\begin{aligned}
\text{P(A$_1=+1$, B$_0=+1$)}&=0.152 \text{b$_{12}$}+0.143 \text{b$_{16}$}+0.113 \text{b$_{20}$}\\&+0.078 \text{b$_{28}$}+0.197 \text{b$_{4}$}+0.205 \text{b$_{48}$}\\
&+0.162 \text{b$_{52}$}+0.125 \text{b$_{60}$}+0.250I,\\
\text{P(A$_1=+1$, B$_0=-1$)}&=-0.152 \text{b$_{12}$}+0.143 \text{b$_{16}$}-0.113 \text{b$_{20}$}\\&-0.078 \text{b$_{28}$}-0.197 \text{b$_{4}$}+0.205 \text{b$_{48}$}\\
&-0.162 \text{b$_{52}$}-0.125 \text{b$_{60}$}+0.250 I,\\
\text{P(A$_1=-1$, B$_0=+1$)}&=0.152 \text{b$_{12}$}-0.143 \text{b$_{16}$}-0.113\text{b$_{20}$}\\&-0.078\text{b$_{28}$}+0.197 \text{b$_{4}$}-0.205\text{b$_{48}$}\\
&-0.162\text{b$_{52}$}-0.125 \text{b$_{60}$}+0.250I,\\
\text{P(A$_1=-1$, B$_0=-1$)}&=-0.152 \text{b$_{12}$}-0.143 \text{b$_{16}$}+0.113 \text{b$_{20}$}\\&+0.078\text{b$_{28}$}-0.197\text{b$_{4}$}-0.205\text{b$_{48}$}\\
&+0.162\text{b$_{52}$}+0.125 \text{b$_{60}$}+250I.
\label{eq:decomposition1}
\end{aligned}
\end{equation}

\begin{equation}
\begin{aligned}
\text{P(A$_1=+1$)}&=0.286 \text{b$_{16}$}+0.410\text{b$_{48}$}+0.500I,\\
\text{P(E$_1=+1$)}&=0.149 \text{b$_{1}$}+0.476 \text{b$_{3}$}+0.500I,\\
\text{P(A$_0=+1$, E$_0=+1$)}&=0.197 \text{b$_{1}$}+0.152 \text{b$_{3}$}+0.250 \text{b$_{48}$}\\&+0.197 \text{b$_{49}$}+0.152 \text{b$_{51}$}+0.250 I,\\
\text{P(A$_0=+1$, E$_0=-1$)}&=-0.197 \text{b$_{1}$}-0.152 \text{b$_{3}$}+0.250 \text{b$_{48}$}\\&-0.197 \text{b$_{49}$}-0.152 \text{b$_{51}$}+0.250 I,\\
\text{P(A$_0=-1$, E$_0=+1$)}&=0.197 \text{b$_{1}$}+0.152\text{b$_{3}$}-0.250 \text{B$_{48}$}\\&-0.197 \text{b$_{49}$}-0.152 \text{b$_{51}$}+0.250 I,\\
\text{P(A$_0=-1$, E$_0=-1$)}&=-0.197 \text{b$_{1}$}-0.152 \text{b$_{3}$}-0.250 \text{b$_{48}$}\\&+0.197\text{b$_{49}$}+0.152 \text{b$_{51}$}+0.250I,\\
\text{P(A$_0=+1$, E$_1=+1$)}&=0.074 \text{b$_{1}$}+0.238 \text{b$_{3}$}+0.250 \text{b$_{48}$}\\&+0.074 \text{b$_{49}$}+0.238 \text{b$_{51}$}+0.250 I,\\
\text{P(A$_0=+1$, E$_1=-1$)}&=-0.074 \text{b$_{1}$}-0.238 \text{b$_{3}$}+0.250\text{b$_{48}$}\\&-0.074 \text{b$_{49}$}-0.238 \text{b$_{51}$}+0.250 I,\\
\text{P(A$_0=-1$, E$_1=+1$)}&=0.074 \text{b$_{1}$}+0.238 \text{b$_{3}$}-0.250\text{b$_{48}$}\\&-0.074 \text{b$_{49}$}-0.238 \text{b$_{51}$}+0.250I,\\
\text{P(A$_0=-1$, E$_1=-1$)}&=-0.074\text{b$_{1}$}-0.238 \text{b$_{3}$}-0.250 \text{b$_{48}$}\\&+0.074 \text{b$_{49}$}+0.238 \text{b$_{51}$}+0.250 I,\\
\end{aligned}
\end{equation}

\begin{equation}
\begin{aligned}
\text{P(E$_1=+1$, A$_1=+1$)}&=0.143 \text{b$_{1}$}+0.074 \text{b$_{16}$}+0.042 \text{b$_{17}$}\\&+0.061 \text{b$_{19}$}+0.205 \text{b$_{3}$}+0.238 \text{b$_{48}$}\\
&+0.136 \text{b$_{49}$}+0.195 \text{B$_{51}$}+0.250 I,\\
\text{P(E$_1=+1$, A$_1=-1$)}&=-0.143 \text{b$_{1}$}+0.074 \text{b$_{16}$}-0.042 \text{b$_{17}$}\\&-0.061 \text{b$_{19}$}-0.205 \text{b$_{3}$}+0.238 \text{b$_{48}$}\\
&-0.136\text{b$_{49}$}-0.195 \text{b$_{51}$}+0.250 I,\\
\text{P(E$_1=-1$, A$_1=+1$)}&=0.1430 \text{b$_{1}$}-0.074 \text{b$_{16}$}-0.042 \text{b$_{17}$}\\&-0.061 \text{b$_{19}$}+0.205 \text{b$_{3}$}-0.238 \text{b$_{48}$}\\
&-0.136 \text{b$_{49}$}-0.195 \text{b$_{51}$}+0.250I,\\
\text{P(E$_1=-1$, A$_1=-1$)}&=-0.143 \text{b$_{1}$}-0.074 \text{b$_{16}$}+0.042 \text{b$_{17}$}\\&+0.061 \text{b$_{19}$}-0.205 \text{b$_{3}$}-0.238 \text{b$_{48}$}\\
&+0.136 \text{b$_{49}$}+0.195 \text{b$_{51}$}+0.250 I,\\
\end{aligned}
\end{equation}

\begin{equation}
\begin{aligned}
\text{P(A$_1=+1$, E$_0=+1$)}&=0.197\text{b$_{1}$}+0.143 \text{b$_{16}$}+0.113\text{b$_{17}$}\\&+0.078 \text{b$_{19}$}+0.152 \text{b$_{3}$}+0.205 \text{b$_{48}$}\\
&+0.162 \text{b$_{49}$}+0.125 \text{b$_{51}$}+0.250 I,\\
\text{P(A$_1=+1$, E$_0=-1$)}&=-0.197  \text{b$_{1}$}+0.143  \text{b$_{16}$}-0.113  \text{b$_{17}$}\\&-0.078  \text{b$_{19}$}-0.152  \text{b$_{3}$}+0.205  \text{b$_{48}$}\\
&-0.162  \text{b$_{49}$}-0.125  \text{b$_{51}$}+0.250 I,\\
\text{P(A$_1=-1$, E$_0=+1$)}&=0.197 \text{b$_{1}$}-0.143 \text{b$_{16}$}-0.113 \text{b$_{17}$}\\&-0.078\text{b$_{19}$}+0.152 \text{b$_{3}$}-0.205 \text{b$_{48}$}\\
&-0.162 \text{b$_{49}$}-0.125 \text{b$_{51}$}+0.250 I,\\
\text{P(A$_1=-1$, E$_0=-1$)}&=-0.197 \text{b$_{1}$}-0.143 \text{b$_{16}$}+0.113 \text{b$_{17}$}\\&+0.078 \text{b$_{19}$}-0.152 \text{b$_{3}$}-0.205 \text{b$_{48}$}\\
&+0.162 \text{b$_{49}$}+0.125 \text{b$_{51}$}+0.250 I.
\label{eq:decomposition2}
\end{aligned}
\end{equation}

where $\text{b}{_i}=\text{tr}(\rho \text{B}{_i})$ with $\text{B}{_i}$ are the
Pauli operators. The order of $\text{B}{_i}$ is in the four-base subscript and
the base-four notation, 0, 1, 2, 3 can be directly mapped to either identity or
Pauli $x$, $y$, and $z$ matrices. For example $\text{B}{_3}$  has the form $I\otimes
I\otimes \sigma_z$ where $ I, \sigma_z$ are identity and Pauli $z$ 
matrices, respectively.
Similarly we can find the forms of other
$\text{B}{_i}$ (details are given
in~\cite{singh-pra-18}).

%

\end{document}